\documentclass[11pt,leqno]{article}

 \usepackage[margin=1.2in]{geometry}
\usepackage{hyperref}
\usepackage{amsfonts}
\usepackage{amsmath,amssymb}
\usepackage{graphicx,amscd,xypic,color}
\usepackage{bbding}
\usepackage{indentfirst}
\usepackage{mathrsfs, dsfont, bbm,enumerate}
\usepackage[T1]{fontenc}

\usepackage[center]{titlesec}
\titleformat*{\section}{\centering}

\newtheorem{thm}{Theorem}

\newtheorem{assmp}{Assumption}
\newtheorem{example}{Example}
\newtheorem{remark}{Remark}
\newenvironment{proof}{\hspace{0ex}\textsc{Proof}.\hspace{1ex}}{\hfill$\Box$\newline}
\DeclareMathOperator{\dd}{\mathrm{d\!}}
\DeclareMathOperator{\BE}{\mathbf{E}}
\DeclareMathOperator{\BP}{\mathbf{P}}

\DeclareMathOperator{\R}{\mathbb{R}}

\DeclareMathOperator{\setq}{\mathcal{Q}}
 \DeclareMathOperator{\setqt}{\widetilde{\setq}}
\DeclareMathOperator{\setg}{\mathcal{G}}

\def\inversew{\nu}

\begin{document}
\title{\textbf{A NOTE ON THE QUANTILE FORMULATION}}
\author{\textsc{Zuo Quan Xu}\footnote{Department of Applied Mathematics, The Hong Kong Polytechnic University, Hong Kong. Email: \url{maxu@polyu.edu.hk}. The author acknowledges financial supports from Hong Kong General Research Fund (No. 529711), Hong Kong Early Career Scheme (No. 533112), and The Hong Kong Polytechnic University. } \\[10pt] \textit{The Hong Kong Polytechnic University}}
\date{07 April 2014}
\maketitle
\begin{abstract}
Many investment models in discrete or continuous-time settings boil down to maximizing an objective of the quantile function of the decision variable. This quantile optimization problem is known as the quantile formulation of the original investment problem. Under certain monotonicity assumptions, several schemes to solve such quantile optimization problems have been proposed in the literature. In this paper, we propose a change-of-variable and relaxation method to solve the quantile optimization problems without using the calculus of variations or making any monotonicity assumptions. The method is demonstrated through a portfolio choice problem under rank-dependent utility theory (RDUT). We show that this problem is equivalent to a classical Merton's portfolio choice problem under expected utility theory with the same utility function but a different pricing kernel explicitly determined by the given pricing kernel and probability weighting function. With this result, the feasibility, well-posedness, attainability and uniqueness issues for the portfolio choice problem under RDUT are solved. It is also shown that solving functional optimization problems may reduce to solving probabilistic optimization problems. The method is applicable to general models with law-invariant preference measures including portfolio choice models under cumulative prospect theory (CPT) or RDUT, Yaari's dual model, Lopes' SP/A model, and optimal stopping models under CPT or RDUT.
\\[1pt]
\end{abstract}

\section{INTRODUCTION}
\noindent
Classical expected utility theory (EUT) as a model of choice under uncertainty fails to explain a number of paradoxes. Among the alternative models proposed, Kahneman and Tversky's (1979, 1992) cumulative prospect theory (CPT) provides one of the best explanations of these paradoxes. This theory consists of three components: an $S$-shaped utility function\footnote{A function is called $S$-shaped if it is convex on the left and concave on the right; and reverse $S$-shaped if concave on the left and convex on the right.}, a reference point, and probability weighting/distortion functions. The last two are missing in EUT. In light of these theoretical developments, it is natural to consider investment problems that involve probability weighting functions. However, the probability weighting functions make these problems time-inconsistent so that these problems cannot be studied using only classical dynamic programming or probabilistic approaches.
\par
Jin and Zhou (2008) initiated the study of portfolio choice problems under CPT with probability weighting functions in continuous-time settings. They solved the problem by assuming the monotonicity of a function related to the pricing kernel and probability weighting function. However, this assumption is so restrictive that it excludes most probability weighting functions that are typically used, including that proposed by Tversky and Kahneman (1992), in the Black-Scholes market setting. Jin, Zhang, and Zhou (2011) considered the same portfolio choice problem under the scenario of a loss constraint with the same assumption. He and Zhou (2011) investigated general models with law-invariant preference measures, including the classical Merton's portfolio choice model under EUT, the mean-variance model, the goal reaching model, the Yaari's dual model, the Lopes' SP/A model, the behavioral model under CPT, and those explicitly involving VaR and CVaR in their objectives and/or constraints. Their work took a step forward and reduced the monotonicity assumption in Jin and Zhou (2008) to a piece-wise monotonicity assumption. The results cover the probability weighting functions proposed by Tversky and Kahneman (1992), Tversky and Fox (1995), and Prelec (1998). Xu and Zhou (2013) initiated the study of continuous-time optimal stopping problem under CPT and solved the problem under the same assumption of piece-wise monotonicity as He and Zhou (2011). By adopting the calculus of variations, Xia and Zhou (2012) achieved a breakthrough. They proposed and solved a portfolio choice problem under rank-dependent utility theory (RDUT) with no monotonicity assumptions. Their method also works for general models with law-invariant preference measures. However, they use techniques from the calculus of variations and have extensive recourse to convex analysis, so their arguments are lengthy, technical, and difficult to follow.
 \par
 In this paper, without making any monotonicity assumptions, we propose a new and easy-to-follow method to study the portfolio choice problem under RDUT. A complete and compact argument replaces the lengthy calculus of variations argument in Xia and Zhou (2012). The main idea is as follows. After transforming the portfolio choice problem into its quantile formulation, we make a change of variable to remove the probability weighting function from the objective and reveal the essence of the problem. In the literature, the optimal solution is commonly obtained by point-wise maximizing the Lagrangian in the objective. However, such a solution may not be a quantile function. Our idea is to replace a part of the Lagrangian to relax the problem so that the new problem can be solved by point-wise maximizing the new Lagrangian, and then to show that there is no gap between the old and new Lagrangians in this point-wise solution. Through this approach, we show that solving a portfolio choice problem under RDUT reduces to solving a classical Merton's portfolio choice problem under EUT with the same utility function but a different pricing kernel, which is determined by the given pricing kernel and probability weighting function. Moreover, the quantile optimization problem is avoided in the latter. As with Xia and Zhou (2012), the method is applicable to general models with law-invariant preference measures.
 \par
 In the literature, there is no study on feasibility, well-posedness, attainability and uniqueness issues for the portfolio choice problem under RDUT\footnote{see, e.g., Jin, Xu and Zhou (2008) for the definitions of feasibility, well-posedness, attainability and uniqueness issues for a portfolio choice problem}. We investigate these issues by linking the portfolio choice problem under RDUT to a classical Merton's portfolio choice problem under EUT for which the issues have been completely solved in Jin, Xu and Zhou (2008).
 \par
The remainder of this paper is organized as follows. In Section 2, we formulate a portfolio choice problem under RDUT and define its quantile formulation. In Section 3, we introduce a key step --- making a change of variable --- to formulate an equivalent quantile optimization problem, in which the probability weighting function is removed from the objective. The problem is then completely solved by a new relaxation method in Section 4. In Section 5, we demonstrate how to transform the portfolio choice problem under RDUT into an equivalent classical Merton's portfolio choice problem under EUT. The feasibility, well-posedness, attainability and uniqueness issues for the portfolio choice problem under RDUT are also investigated in this section. We conclude the paper in Section 6.

\section{\textsc{PROBLEM FORMULATION}}
\noindent
Using martingale representation theory (see, e.g., Pliska (1986), Karatzas, Lehoczky, and Shreve (1987), Cox and Huang (1989, 1991)),
the dynamic portfolio choice problem under RDUT in a complete market setting\footnote{See, e.g., Xia and Zhou (2012).} reduces to finding a random outcome $X$ to
\begin{align}\label{objective}
\sup_X \quad & \int_0^{\infty} u(x)\dd\:\big(1-w(1-F_X(x))\big),\\
\nonumber\textrm{subject to}\quad & \BE[ {\rho} X]=x_0,\quad X\geqslant 0,
\end{align}
where $F_X(\cdot)$ is the probability distribution function of $X$; $w(\cdot)$ is the probability weighting function which is differentiable and strictly increasing on $[0,1]$ with $w(0)=0$ and $w(1)=1$; $u(\cdot)$ is the utility function which is strictly increasing and second order differentiable on $\R^+$ with $u''(\cdot)<0$; and $\rho>0$ is the pricing kernel, also called the stochastic discount factor or state pricing density. We always have that $\BE[\rho]<+\infty$.
\par
If $w(\cdot)$ is the identity function, i.e., $w(x)=x$ for all $x\in[0,1]$, then
\begin{align*}
 \int_0^{\infty} u(x)\dd\:\big(1-w(1-F_X(x))\big)=\int_0^{\infty} u(x)\dd F_X(x)=\BE[u(X)],
\end{align*}
for any $X\geqslant 0$, and consequently, problem \eqref{objective} reduces to a classical Merton's portfolio choice problem under EUT:
\begin{align*}
\sup_X \quad & \BE[u(X)],\\
\textrm{subject to}\quad & \BE[ {\rho} X]=x_0,\quad X\geqslant 0.
\end{align*}
\par
To tackle problem \eqref{objective}, in the literature (see, e.g., Jin and Zhou (2008), Jin, Zhang, and Zhou (2011), He and Zhou (2011, 2012), Xia and Zhou (2012)), it is always assumed that
\begin{assmp}\label{assmp:atomless}
The pricing kernel is atomless\footnote{A random variable is called atomless or non-atomic if its cumulative distribution function is continuous, and called atomic otherwise.}.
\end{assmp}
Under this assumption, solving problem \eqref{objective} then reduces to solving a quantile\footnote{The quantile function $Q(\cdot)$ of a real-valued random variable is defined as the right-continuous inverse function of its cumulative distribution function $F(\cdot)$, that is $Q(x)=\sup\{t\in\R: F(t)\leqslant x\}$, for all $x\in (0,1)$, with convention $\sup\emptyset=-\infty$. A real-valued random variable is atomless if and only if its quantile function is strictly increasing.} optimization problem
\begin{align} \label{objective0}
\sup\limits_{G(\cdot)\in\setg_{x_0}} \int_{0}^{1}u(G(x))w'(1-x) \dd x,
\end{align}
where the set $\setg_{x_0}$ is given by
\begin{align*}
\setg_{x_0}&:=\left\{G(\cdot)\in\setg :\int_0^1 G(x)F^{-1}_{\rho}(1-x)\dd x=x_0\right\},
\end{align*}
 the set $\setg$ denotes the set of all quantile functions:
\begin{align*}
\setg&:=\left\{ G(\cdot): (0,1)\mapsto\R^+, \text{ increasing and right-continuous with left limits (RCLL)} \right\},
\end{align*}
and $F^{-1}_{\rho}(\cdot) \in\setg$ denotes the quantile function of the pricing kernel $\rho$. By Assumption \ref{assmp:atomless}, $\rho$ is atomless, so $F^{-1}_{\rho}(\cdot)$ is strictly increasing.
\par
Problem \eqref{objective} and problem \eqref{objective0} are linked as follows. The optimal solution $X^*$ to problem \eqref{objective} and the optimal solution $G^*(\cdot)$ to problem \eqref{objective0} satisfy
\begin{align}\label{xstar0}
 X^*=G^*(1-F_{\rho}(\rho)).
\end{align}
For this reason, problem \eqref{objective0} is called the quantile formulation of problem \eqref{objective}.
\par
Before Xia and Zhou (2012), problem \eqref{objective0} was partially solved under certain monotonicity assumptions in the literature.
Xia and Zhou (2012) used the calculus of variations to tackle it without making those monotonicity assumptions, but their arguments are lengthy and complex. Moreover, they did not study the feasibility, well-posedness, attainability or uniqueness issues for problem \eqref{objective}.
\par
In this paper, we propose a simple change-of-variable and relaxation method to tackle problem \eqref{objective0} without making any assumptions. We also solve the feasibility, well-posedness, attainability and uniqueness issues for problem \eqref{objective} by linking the problem to a classical Merton's portfolio choice problem under EUT.
\begin{remark}
In the literature, $\setg_{x_0}$ is often replaced by
\begin{align*}
\overline{\setg}_{x_0}&:=\left\{G(\cdot)\in\setg :\int_0^1 G(x)F^{-1}_{\rho}(1-x)\dd x\leqslant x_0\right\}.
\end{align*}
However, there is no difference between considering problem \eqref{objective0} for $\setg_{x_0}$ or $\overline{\setg}_{x_0}$ because the optimal solution to problem \eqref{objective0} in $\overline{\setg}_{x_0}$, if it exists, must belong to $\setg_{x_0}$.
\end{remark}
\begin{remark}
Here we assume that the pricing kernel is atomless as according to convention. However, if one studies economic equilibrium models with law-invariant preference measures (see, e.g., Xia and Zhou (2012)), the pricing kernel will be a part of the solution, so one cannot make a priori any assumption on it. The quantile formulation problem with an atomic pricing kernel is solved in Xu (2014).
\end{remark}

\section{\textsc{CHANGE OF VARIABLE}}
 \noindent
 To tackle problem \eqref{objective0}, our first main idea in this paper is to make a change of variable to remove the probability weighting function from the objective.
 \par
Let $\inversew: [0,1]\mapsto [0,1]$ be the inverse mapping of $x\mapsto 1-w(1-x)$, that is
\[\inversew(x):=1-w^{-1}(1-x), \quad x\in[0,1].\]
Then $\inversew(\cdot)$ is also a probability weighting function that is differentiable and strictly increasing on $[0,1]$. It follows that
\begin{multline*}
\int_{0}^{1}u(G(x))w'(1-x) \dd x=\int_{0}^{1}u(G(x))\dd\; (1-w(1-x))\\
=\int_{0}^{1}u(G(x))\dd\; (\inversew^{-1}(x))=\int_{0}^{1}u(G(\inversew(x)))\dd x=\int_{0}^{1}u(Q(x))\dd x,
\end{multline*}
where \[ Q(x)=G(\inversew(x)),\quad x\in(0,1).\] Note that
\begin{multline*}
\setg_{x_0}=\left\{G(\cdot)\in\setg:\int_0^1 G(x)F^{-1}_{\rho}(1-x)\dd x=x_0 \right\}\\
=\left\{G(\cdot)\in\setg :\int_0^1 G(\inversew(x))F^{-1}_{\rho}(1-\inversew(x))\inversew'(x)\dd x=x_0 \right\}.
\end{multline*}
Therefore, we conclude that $G(\cdot)\in\setg_{x_0}$ if and only if $Q(\cdot)\in\setq$, where
\begin{multline*}
\setq:=\left\{Q(\cdot): (0,1)\mapsto\R^+,\text{ increasing and RCLL with }\int_0^1 Q(x)\varphi'(x)\dd x=x_0\right\}\\
=\left\{Q(\cdot)\in \setg:\int_0^1 Q(x)\varphi'(x)\dd x=x_0\right\},
\end{multline*}
and
\begin{multline}\label{defi:variphi}
\varphi(x):=-\int_x^1 F^{-1}_{\rho}(1-\inversew(y))\inversew'(y)\dd y=-\int_{\inversew(x)}^1F^{-1}_{\rho}(1-y)\dd y \\
=-\int_0^{1-\inversew(x)}F^{-1}_{\rho}(y)\dd y=-\int_{0}^{w^{-1}(1-x)}F^{-1}_{\rho}(y)\dd y , \quad x\in[0,1].
\end{multline}
Note that $\varphi(\cdot)$ is a differentiable and strictly increasing function on $[0,1]$ with $\varphi(0)=-\BE[\rho]$ and $\varphi(1)=0$.
\par
By making this change of variable, problem \eqref{objective0} has now been transformed into an equivalent problem:
 \begin{align} \label{objective0'}
\sup\limits_{Q(\cdot)\in\setq} \int_{0}^{1}u(Q(x)) \dd x,
\end{align}
in which the probability weighting function does not appear in the objective. From now on, we focus on this problem.
\par
We point out here that although the objective of problem \eqref{objective0'} does not involve the probability weighting function, the constraint set $\setq$ does. So problem \eqref{objective0'} is different from the special scenario of problem \eqref{objective0}, in which $w(\cdot)$ is replaced by the identity function. We will study their relationship in Section 5.
 \par
This change in the formulation of problem \eqref{objective0} is mathematically simple, but reveals the essence of the problem. In problem \eqref{objective0'}, the function $\varphi(\cdot)$, rather than the probability weighting function and the quantile function of the pricing kernel, plays a key role; whereas, in problem \eqref{objective0}, the probability weighting function and the quantile function of the pricing kernel play separate roles in the objective and the constraint. Because the probability weighting function does not appear in the objective of problem \eqref{objective0'}, we can solve it by a new relaxation approach. Moreover, this also suggests that it may be possible to link problem \eqref{objective0'} to a problem under EUT. This will be investigated after solving it.
\par
We also point out here that the new formulation explains why the function $\varphi'(\cdot)$ plays such an important role in many existing models, such as those introduced by Jin and Zhou (2008), He and Zhou (2011), and Xia and Zhou (2012). In those works, the mysterious function $\varphi'(\cdot)$ is derived after lengthy analysis, and an explanation of why it should appear and play the key role is never provided.
\par
In tackling problem \eqref{objective0}, some studies assume $\varphi(\cdot)$ to satisfy various properties which are not generally true in practice, and under these assumptions, the problem is partially solved. Here are some examples.
\begin{example}
In Jin and Zhou (2008), the function $ \frac{F^{-1}_{\rho}(\cdot)}{w'(\cdot)}$ is assumed to be increasing in Assumption 4.1. This is equivalent to $\varphi'(\cdot)$ being decreasing, i.e., $\varphi(\cdot)$ is a concave function. In fact, we have \[1-w(1-\inversew(x))=x,\quad x\in[0,1],\] so \[\inversew'(x)=\frac{1}{w'(1-\inversew(x))}, \quad x\in[0,1].\] And consequently, by \eqref{defi:variphi},
\begin{align}\label{phi-rho-w}
\varphi'(x)=F^{-1}_{\rho}(1-\inversew(x))\inversew'(x)=\frac{F^{-1}_{\rho}(1-\inversew(x))}{w'(1-\inversew(x))},\quad x\in[0,1].
\end{align}
The equivalence follows immediately as $\inversew(\cdot)$ is increasing.
\end{example}
\begin{example}
In He and Zhou (2011), the function $\frac{w'(1-\cdot )}{F^{-1}_{\rho}(1-\cdot)}$ is assumed to be first strictly increasing and then strictly decreasing in Assumption 3.5 and many of the following results. By \eqref{phi-rho-w}, this is equivalent to $\varphi'(\cdot)$ being first strictly decreasing and then strictly increasing, i.e., $\varphi(\cdot)$ is a strictly reverse $S$-shaped function.
\end{example}
\begin{example}
In He and Zhou (2012), the function $\frac{w'(1-\cdot )}{F^{-1}_{\rho}(1-\cdot)}$ is assumed to be nondecreasing in Theorem 2, which is equivalent to $\varphi'(\cdot)$ being decreasing, i.e., $\varphi(\cdot)$ is a concave function. In Proposition 4-7, Theorem 4-6, and Corollary 1, the same function $\frac{w'(1-\cdot )}{F^{-1}_{\rho}(1-\cdot)}$ is assumed to be first strictly decreasing and then strictly increasing. This is equivalent to $\varphi'(\cdot)$ being first strictly increasing and then strictly decreasing, i.e., $\varphi(\cdot)$ is a strictly $S$-shaped function.
\end{example}

\section{\textsc{A NEW RELAXATION APPROACH}}
\noindent
Our second main idea in this paper is to introduce a simple relaxation method to tackle problem \eqref{objective0'}.
\par
The objective of problem \eqref{objective0'} is concave with respect to the decision quantiles, so we can apply the Lagrange multiplier method.
Problem \eqref{objective0'} is equivalent to problem
\begin{align} \label{objective1}
\sup\limits_{Q(\cdot) \in\setg}J(Q (\cdot))&=\int_{0}^{1}\Big(u(Q(x))-\lambda Q(x)\varphi'(x)\Big)\dd x,
\end{align}
for some Lagrange multiplier $\lambda> 0$ in the sense that they admit the same optimal solution.
\par
A naive approach to tackling the foregoing problem \eqref{objective1} is to point-wise maximize its Lagrangian (the integrand in \eqref{objective1}) to get a point-wise solution
\[Q_0(x):=\arg\max\Big\{y:u(y)-\lambda y\varphi'(x)\Big\}=(u')^{-1}(\lambda\varphi'(x)),\quad x\in(0,1).\]
However, this point-wise solution may not be a quantile function in $\setg$. In fact, $Q_0(\cdot)$ is a quantile function if and only if it is increasing, that is equivalent to $\varphi(\cdot)$ being concave. This is exactly what has been assumed in Jin and Zhou (2008) so as to solve the problem.
\par
The novel idea in this paper is to replace $\varphi(\cdot)$ by some function $\delta(\cdot)$ in the Lagrangian of problem \eqref{objective1} so that:
\begin{enumerate}[(i)]
	 \item The new cost function gives an upper bound to that in \eqref{objective1};
 \item The new problem can be solved by point-wise maximizing the new Lagrangian; and
 \item There is no gap between the new and old cost functions in the point-wise solution.
\end{enumerate}
This approach allows us to solve the problem completely without making any assumptions on the function $\varphi(\cdot)$.
\par
We first need to find a relaxed cost function. To this end, let $\delta(\cdot)$ be an absolutely continuous function such that
\begin{align}\label{deltarequirement}
 \int_{0}^{1}\Big(u(Q(x))-\lambda Q(x)\varphi'(x)\Big)\dd x\leqslant \int_{0}^{1}\Big(u(Q(x))-\lambda Q(x)\delta'(x)\Big)\dd x,
\end{align}
for every $Q(\cdot) \in\setg$. Setting $\delta(0)=\varphi(0)$ and $\delta(1)=\varphi(1)$ and applying Fubini's theorem, the inequality \eqref{deltarequirement} is equivalent to
\begin{align}\label{determinedelta1}
\int_0^1\Big(\varphi(x)-\delta(x)\Big)\dd Q(x)\leqslant 0,
\end{align}
for every $Q(\cdot) \in\setg$, which is clearly equivalent to $\delta(\cdot)$ dominating $\varphi(\cdot)$ on $[0,1]$.
\par
In this case, we have
\begin{multline}\label{keyineq}
 \int_{0}^{1}\Big(u(Q(x))-\lambda Q(x)\varphi'(x)\Big)\dd x\leqslant \int_{0}^{1}\Big(u(Q(x))-\lambda Q(x) \delta' (x)\Big)\dd x\\
\leqslant \int_{0}^{1}\Big(u(\overline{Q}(x))-\lambda \overline{Q}(x) \delta' (x)\Big)\dd x,
\end{multline}
where the last inequality is obtained by point-wise maximizing the new Lagrangian:
\begin{align}\label{overlineQ}
\overline{Q}(x):=\arg\max\Big\{y:u(y)-\lambda y\delta'(x)\Big\}=(u')^{-1}(\lambda \delta' (x)),\quad x\in[0,1].
\end{align}
To make $ \overline{Q}(\cdot)$ a quantile function, we require $\delta(\cdot)$ to be concave.
\par
To make $\overline{Q}(\cdot)$ an optimal solution to problem \eqref{objective1}, it is sufficient, by \eqref{keyineq}, to have
\begin{align}\label{optimal:barQ2}
 \int_{0}^{1}\Big(u(\overline{Q}(x))-\lambda \overline{Q}(x)\varphi'(x)\Big)\dd x=\int_{0}^{1}\Big(u(\overline{Q}(x))-\lambda \overline{Q}(x) \delta' (x)\Big)\dd x,
\end{align}
or equivalently,
 \begin{align*}
 \int_{0}^{1} (u')^{-1}(\lambda \delta' (x))\Big(\varphi'(x)-\delta' (x)\Big)\dd x=0.
\end{align*}
Applying Fubini's theorem and using $\delta(0)=\varphi(0)$ and $\delta(1)=\varphi(1)$, the above identity is equivalent to
 \begin{multline} \label{optimal:barQ}
 \int_{0}^{1} (u')^{-1}(\lambda \delta' (x)) \Big( \varphi'(x)-\delta' (x)\Big)\dd x=\int_{0}^{1} \Big(\delta(x)-\varphi(x)\Big)\dd\; \Big( (u')^{-1}(\lambda \delta' (x))\Big)\\
=\lambda\int_{0}^{1} \Big(\delta(x)-\varphi(x)\Big)\frac{1}{u''\Big((u')^{-1}(\lambda \delta' (x))\Big)}\dd \delta' (x)=0.
\end{multline}
Since $\delta(\cdot)$ dominates $\varphi(\cdot)$ on $[0,1]$, $u''(\cdot)<0$, and $\delta(\cdot)$ is concave, by the last identity, $\delta'(\cdot)$ must be constant on any sub interval of $\{x\in[0,1]: \delta(x)>\varphi(x)\}$.
\par
Putting all of the requirements on $\delta(\cdot)$ obtained thus far together, we see that $\delta(\cdot)$ should
\begin{enumerate}[(i)]
	 \item dominate $\varphi(\cdot)$ on $[0,1]$ with $\delta(0)=\varphi(0)$ and $\delta(1)=\varphi(1)$;
 \item be concave on $[0,1]$; and
 \item be affine on $\{x\in[0,1]: \delta(x)>\varphi(x)\}$.
\end{enumerate}
Therefore, we conclude that $\delta(\cdot)$ must be the concave envelope of $\varphi(\cdot)$ on $[0,1]$:
\begin{align}\label{envelope}
\delta(x)=\sup\limits_{0\leqslant a\leqslant x\leqslant b\leqslant 1}\frac{(b-x)\varphi(a)+(x-a)\varphi(b)}{b-a},\quad x\in[0,1].
\end{align}
\par
On the other hand, if $\delta(\cdot)$ is the concave envelope of $\varphi(\cdot)$ on $[0,1]$, then \eqref{determinedelta1} and \eqref{optimal:barQ} hold true. This further implies, by \eqref{keyineq} and \eqref{optimal:barQ2}, that $\overline{Q}(\cdot)$ defined in \eqref{overlineQ} is an optimal solution to problem \eqref{objective1}.
\par
Putting all of the results obtained thus far together and noting that $u(\cdot)$ is strictly concave, we conclude that
\begin{thm}\label{maintheorem}
Problem \eqref{objective1} admits a unique optimal solution
\begin{align*}
 (u')^{-1}(\lambda \delta' (x)),\quad x\in(0,1),
 \end{align*}
 where $\delta(\cdot)$ defined in \eqref{envelope} is the concave envelope of $\varphi(\cdot)$ on $[0,1]$.
 \par
Problem \eqref{objective0'} admits an optimal solution if and only if
 \begin{align*}
 \int_{0}^{1} (u')^{-1}(\lambda \delta' (x))\varphi'(x)\dd x=x_0
\end{align*}
admits a solution $\lambda>0$, in which case
\begin{align*}
 (u')^{-1}(\lambda \delta' (x)), \quad x\in(0,1),
 \end{align*}
is the unique optimal solution to problem \eqref{objective0'}.
\end{thm}
\begin{proof}
The foregoing argument shows that \[ (u')^{-1}(\lambda \delta' (x)),\quad x\in(0,1),\] is an optimal solution to problem \eqref{objective1}. Since $u(\cdot)$ is strictly concave, the optimal solution is unique.
\par
Suppose problem \eqref{objective0'} admits an optimal solution. Then the solution must be an optimal solution to problem \eqref{objective1} for some $\lambda>0$, so it must be of the form \[ (u')^{-1}(\lambda \delta' (x)),\quad x\in(0,1).\] This should be a feasible solution to problem \eqref{objective0'}, so \[\int_{0}^{1} (u')^{-1}(\lambda \delta' (x))\varphi'(x)\dd x=x_0.\]
\par
On the other hand, suppose that \[ \int_{0}^{1} (u')^{-1}(\lambda \delta' (x))\varphi'(x)\dd x=x_0\] holds true for some $\lambda>0$. Note that
\[\int_{0}^{1} Q(x)\varphi'(x)\dd x=x_0\] for all $Q(\cdot) \in\setq$, so
\begin{multline*}
\sup\limits_{Q(\cdot) \in\setq} \int_{0}^{1}u(Q(x)) \dd x=
\sup\limits_{Q(\cdot) \in\setq} \int_{0}^{1}\Big(u(Q(x))-\lambda Q(x)\varphi'(x)\Big)\dd x+\lambda x_0\\
\leqslant
\sup\limits_{Q(\cdot) \in\setg}\int_{0}^{1}\Big(u(Q(x))-\lambda Q(x)\varphi'(x)\Big)\dd x+\lambda x_0,
\end{multline*}
where the last inequality is due to $\setq\subseteq \setg$.
The optimization problem on the right-hand side is nothing but problem \eqref{objective1}, so the unique solution is\[ (u')^{-1}(\lambda \delta' (x)),\quad x\in(0,1).\] This solution belongs to $\setq$ as $ \int_{0}^{1} (u')^{-1}(\lambda \delta' (x))\varphi'(x)\dd x=x_0$, so it is a feasible solution to the problem on the left-hand side, and consequently, it is an optimal solution to problem \eqref{objective0'}. Since $u(\cdot)$ is strictly concave, the optimal solution to problem \eqref{objective0'} is unique. The proof is complete.
\end{proof}
\par
By Theorem \ref{maintheorem}, the optimal solution to problem \eqref{objective0} is given by
$$G^*(x)=(u')^{-1}(\lambda \delta' (\inversew^{-1}(x)))=(u')^{-1}(\lambda \delta' ( 1-w(1-x))),\quad x\in(0,1),$$
which is the same as the last identity on page 14 in Xia and Zhou (2012). That is, our approach yields the same result as in Xia and Zhou (2012). It is clear that our change-of-variable and relaxation approach is much simpler and neater than the calculus of variations approach in Xia and Zhou (2012), which has extensive recourse to convex analysis. If $\varphi(\cdot)$ is assumed to take special shape, such as reverse $S$-shaped function in He and Zhou (2011), $S$-shaped function in He and Zhou (2012), then we can get explicit expression for $\delta(\cdot)$, and consequently, $G^*(\cdot)$ reduces to the results obtained in those works.

\par
The feasibility, well-posedness, attainability and uniqueness issues for problem \eqref{objective} are very important and hard to answer.
To avoid these issues, various assumptions are used in the literature to ensure the existence and uniqueness of solutions (see, e.g., Jin and Zhou (2008), Jin, Zhang, and Zhou (2011), He and Zhou (2011, 2012)). In the following section, with Theorem \ref{maintheorem}, we will link problem \eqref{objective} to a classical Merton's portfolio choice problem under EUT, for which the feasibility, well-posedness, attainability and uniqueness issues are studied in Jin, Xu, and Zhou (2008). This connection also develops a new way to solve problem \eqref{objective}, which avoids dealing with the quantile formulation problem \eqref{objective0}.

\section{\textsc{A LINK BETWEEN MODELS UNDER RDUT AND EUT}}
\noindent
By Theorem \ref{maintheorem}, it is clear that a quantile function is an optimal solution to problem \eqref{objective0'} if and only if it is an optimal solution to problem
 \begin{align} \label{objective2}
\sup\limits_{Q(\cdot)\in \setqt } \int_{0}^{1}u(Q(x)) \dd x,
\end{align}
where
\begin{align*}
\setqt:=\left\{Q(\cdot)\in \setg:\int_0^1 Q(x)\delta'(x)\dd x=x_0\right\}.
\end{align*}
\par
Since $\delta'(\cdot)$ is decreasing, function
\[F_{\widetilde{\rho}}^{-1}(x):=\delta'(1-x),\quad x\in(0,1),\]
belongs to $\setg$ and can be regarded as the quantile function of some positive random variable $\widetilde{\rho}$. It is possible to choose $\widetilde{\rho}$ to be comonotonic\footnote{Two random variables $X$ and $Y$ are said to be comonotonic if $(X(\omega')-X(\omega)) (Y(\omega')-Y(\omega))\geqslant 0$ almost surely under $\BP\otimes\BP$.} with $\rho$, which is henceforth assumed.\footnote{In fact, $\widetilde{\rho}=\delta'(1-F_{\rho}(\rho))$ in the current setting. Xu (2014) proved that $\widetilde{\rho}$ can be chosen to be comonotonic with $\rho$ even if $\rho$ is not atomless.} Then
 \begin{align*}
\setqt=&\left\{Q(\cdot)\in \setg:\int_0^1 Q(x)\delta'(x)\dd x=x_0\right\}\\
=&\left\{Q(\cdot)\in \setg:\int_0^1 Q(x)F_{\widetilde{\rho}}^{-1}(1-x) \dd x=x_0\right\}.
\end{align*}
\par
Now, we see that problem \eqref{objective2} can be regarded as a special case of problem \eqref{objective0}, in which the probability weighting function $w(\cdot)$ is replaced by the identity function and the pricing kernel $\rho$ is replaced by $\widetilde{\rho}$.
\par
We point out here that the new pricing kernel $\widetilde{\rho}$ may be atomic, which does not satisfy Assumption \ref{assmp:atomless}. In fact, $ \widetilde{\rho}$ is atomless if and only if its quantile function $F_{\widetilde{\rho}}^{-1}(\cdot)$ is strictly increasing. This is equivalent to $\delta(\cdot)$ being strictly concave as $F_{\widetilde{\rho}}^{-1}(\cdot)=\delta'(1-\cdot)$, and also equivalent to $\varphi(\cdot)$ being strictly concave as $\delta(\cdot)$ is the concave envelope of $\varphi(\cdot)$.
\par
Recalling the relationship between problem \eqref{objective} and problem \eqref{objective0}, it is natural to link problem \eqref{objective2} to a portfolio choice problem
\begin{align*}
\sup_X \quad & \int_0^{\infty} u(x) \dd F_X(x),\\
\nonumber\textrm{subject to}\quad & \BE[ \widetilde{\rho} X]=x_0,\quad X\geqslant 0.
\end{align*}
Note that \[\int_0^{\infty} u(x) \dd F_X(x)=\BE[u(X)],\] for any $X\geqslant 0$, so the above problem is the same as problem
\begin{align} \label{equivalentEUT}
\sup_X \quad & \BE[u(X)],\\
\nonumber\textrm{subject to}\quad & \BE[ \widetilde{\rho} X]=x_0,\quad X\geqslant 0.
\end{align}
This is a classical Merton's portfolio choice problem under EUT.
\par
Under the assumption that $\rho$ is atomless, we have linked problem \eqref{objective0} to problem \eqref{objective}. However, we cannot directly link problem \eqref{objective2} to problem \eqref{equivalentEUT} as before, because the new pricing kernel $ \widetilde{\rho}$ in problem \eqref{equivalentEUT} may not be atomless.
\par
 The following result from Xu (2014), where no assumption on $\widetilde{\rho}$ is required, links problem \eqref{objective2} to problem \eqref{equivalentEUT}.
\begin{thm}\label{noatem}
If $\widetilde{X}^*$ is an optimal solution to problem \eqref{equivalentEUT}, then its quantile function is an optimal solution to problem \eqref{objective2}.
\par
On the other hand, if $\widetilde{Q}^*(\cdot)$ is an optimal solution to problem \eqref{objective2}, then \[\widetilde{X}^*:=\widetilde{Q}^*(1-U)\] is an optimal solution to problem \eqref{equivalentEUT}, where $U$ is any random variable uniformly distributed on the unit interval $(0,1)$ and comonotonic with $\widetilde{\rho}$.
\end{thm}
With this result, we can link problem \eqref{equivalentEUT} to problem \eqref{objective}.
\begin{thm}
Let $\widetilde{X}^*$ be an optimal solution to problem \eqref{equivalentEUT} and $\widetilde{Q}^*(\cdot)$ be its quantile function. Then
\begin{align*}
 X^*:=\widetilde{Q}^*(1-w(F_{\rho}(\rho)))
\end{align*}
is an optimal solution to problem \eqref{objective}.
\par
On the other hand, if $X^*$ is an optimal solution to problem \eqref{objective},
then there exists a unique quantile function $\widetilde{Q}^*(\cdot)$ such that \[X^*=\widetilde{Q}^*(1-w(F_{\rho}(\rho))).\] Moreover, $\widetilde{Q}^*(1-U)$ is an optimal solution to problem \eqref{equivalentEUT}, where $U$ is any random variable uniformly distributed on the unit interval $(0,1)$ and comonotonic with $\widetilde{\rho}$.
\end{thm}
\begin{proof}
Suppose that $\widetilde{X}^*$ is an optimal solution to problem \eqref{equivalentEUT} and $\widetilde{Q}^*(\cdot)$ is its quantile function. By Theorem \ref{noatem}, $\widetilde{Q}^*(\cdot)$ is an optimal solution to problem \eqref{objective2} and problem \eqref{objective0'}. Consequently, \[G^*(x):=\widetilde{Q}^*(\inversew^{-1}(x)),\quad x\in(0,1),\] is an optimal solution to problem \eqref{objective0}. Hence, by \eqref{xstar0},
\begin{align*}
 X^*=G^*(1-F_{\rho}(\rho))=\widetilde{Q}^*(\inversew^{-1}(1-F_{\rho}(\rho)))=\widetilde{Q}^*(1-w(F_{\rho}(\rho)))
\end{align*}
is an optimal solution to problem \eqref{objective}.
\par
On the other hand, if $X^*$ is an optimal solution to problem \eqref{objective}. Then by \eqref{xstar0}, \[X^*=G^*(1-F_{\rho}(\rho)),\] where $G^*(\cdot)$ is an optimal solution to problem \eqref{objective0}. Consequently, \[\widetilde{Q}^*( x):=G^*(\inversew(x)),\quad x\in(0,1),\] is an optimal solution to problem \eqref{objective0'} and problem \eqref{objective2}. By Theorem \ref{noatem}, $\widetilde{Q}^*(1-U)$ is an optimal solution to problem \eqref{equivalentEUT}. The proof is complete.
\end{proof}
\par
The above result shows that solving the portfolio choice problem \eqref{objective} under RDUT is equivalent to solving problem \eqref{equivalentEUT} under EUT, which is much easier than the former. Moreover, the latter does not require solving a quantile optimization problem. This provides us a new way to solve the portfolio choice problem \eqref{objective}.
\par
In the literature, various conditions are assumed so as to avoid studying the feasibility, well-posedness, attainability or uniqueness issues for problem \eqref{objective} (see, e.g., Jin and Zhou (2008), Jin, Zhang, and Zhou (2011), He and Zhou (2011, 2012)).
By the above result, these issues for problem \eqref{objective} reduce to that for problem \eqref{equivalentEUT}. However, these issues for problem \eqref{equivalentEUT} are solved in Jin, Xu, and Zhou (2008), so are for problem \eqref{objective}. Similarly, these issues for problems \eqref{objective0}, \eqref{objective0'} and \eqref{objective2} are solved as well.
\begin{remark}
The optimal solution to problem \eqref{equivalentEUT} can be obtained by the Lagrange multiplier method directly. Consequently, its quantile function can be obtained without solving problem \eqref{objective2}. Such approach to solving an investment problem under RDUT without using quantile optimization technique has never appeared in the literature to the best of our knowledge.
\par
On the other hand, this result also tells us that a functional optimization problem \eqref{objective0} can be solved via solving a probabilistic optimization problem \eqref{equivalentEUT}. It is an important and challenging question whether we can apply this idea to other functional optimization problems.
\end{remark}
\begin{remark}
The new pricing kernel $\widetilde{\rho}$ does not depend on the utility function $u(\cdot)$.
\end{remark}
\begin{remark}
Problem \eqref{objective} is time-inconsistent, whereas problem \eqref{equivalentEUT} is time-consistent. It would be interesting to study their relationships as time changes.
\end{remark}

\section{\textsc{CONCLUDING REMARKS}}
\noindent
In this paper, we consider a portfolio choice problem under RDUT. We propose a short, neat, and easy-to-follow method to solve the problem. The method consists of two key ideas. The first is making a change of variable to reveal the key function that we need to consider in the quantile formulation problem. The second is relaxing the Lagrangian so as to find an achievable upper bound. Our approach can also be adopted to deal with portfolio choice and optimal stopping problems under CPT/RDUT as well as many other models with law-invariant preference measures.
\par
The second contribution of this paper is showing that solving a portfolio choice problem under RDUT is equivalent to solving a classical Merton's portfolio choice problem under EUT. The latter avoids studying the quantile optimization problem and can be solved by the classical dynamic programming and probabilistic approaches. Theorem \ref{noatem} obtained by Xu (2014) plays a key role in connecting these two problems as the new pricing kernel cannot be assumed to be atomless in general.
\par
The third contribution of this paper is solving the feasibility, well-posedness, attainability and uniqueness issues for the portfolio choice problem under RDUT.
 \par
Last but not least, we show that solving functional optimization problems may reduce to solving probabilistic optimization problems. This idea may be applicable to other functional optimization problems.
\par
\textsc{Acknowledgments.} The author is grateful to the editors and anonymous referees for carefully reading the manuscript and making useful suggestions that have led to a much improved version of the paper.



\end{document}